\documentclass[conference]{IEEEtran}
\usepackage{cite}
\usepackage{amsthm}

\newtheorem{theorem}{Theorem}

\usepackage{multirow}
\usepackage{braket}
\usepackage{lineno,xcolor}
\usepackage{amsmath}
\usepackage{amsfonts}%
\usepackage{amssymb}
\usepackage{wasysym}
\usepackage{stmaryrd}
\usepackage{bbm} 
\usepackage{tabularx}
\usepackage{multirow}
\usepackage{mathtools}
\usepackage{enumerate}
\usepackage{subfigure}
\usepackage{pgf}
\usepackage{bm}
\usepackage{multicol}
\usepackage{color}
\usepackage{url}
\usepackage{bibunits}
\usepackage{wrapfig}
\usepackage{sidecap}
\usepackage{soul,xcolor}
\usepackage{eurosym}
\usepackage{mathrsfs}
\usepackage[utf8]{inputenc}
\usepackage{algorithmicx}
\usepackage{algpseudocode}
\usepackage{algorithm}
\usepackage{stmaryrd}
\usepackage{upgreek}
\usepackage[hidelinks]{hyperref}
\usepackage{paralist}
\usepackage{graphicx}        
\usepackage{setspace}
\usepackage{tikz}
\usetikzlibrary{arrows}
\usepackage{cleveref}
\usepackage{subfigure}

\usepackage{etoolbox}
\makeatletter
\patchcmd{\@makecaption}
{\scshape}
{}
{}
{}
\makeatother

\DeclareMathAlphabet{\mathpzc}{OT1}{pzc}{m}{it}

\usepackage{epstopdf}
\usepackage{soul}
\setstcolor{red} 
\setul{0pt}{0.7pt} 

\IEEEoverridecommandlockouts

\usepackage{enumitem}
\setlist{nolistsep}
\newcommand*\xor{\oplus}

\begin{document}
\title{Enhanced Min-Sum Decoding of Quantum Codes \\Using Previous Iteration Dynamics}
\author{Dimitris~Chytas, Nithin Raveendran,~\IEEEmembership{Member,~IEEE,} and Bane~Vasi\'{c},~\IEEEmembership{Fellow,~IEEE}\\
\IEEEauthorblockA{Department of Electrical and Computer Engineering The University of Arizona, Tucson, AZ 85721, USA}
Email: \{dchytas, nithin\}@arizona.edu, vasic@ece.arizona.edu
\thanks{This work is supported by the NSF under grants CIF- 2420424, CIF-1855879, CIF-2106189, CCF-2100013, ECCS/CCSS-2027844, ECCS/CCSS-2052751, and in part by Jet Propulsion Laboratory, California Institute of Technology, under a contract with the National Aeronautics and Space Administration and funded through JPL’s Strategic University Research Partnerships (SURP) program. Bane Vasi\'{c} has disclosed an outside interest in his startup company, Codelucida to The University of Arizona. Conflicts of interest resulting from this interest are being managed by The University of Arizona in accordance with its policies. Authors would also like to thank Michele Pacenti for his thoughtful comments and suggestions during the revision of this work. 
}
}

\maketitle

\begin{abstract}
In this paper, we propose a novel message-passing decoding approach that leverages the degeneracy of quantum low-density parity-check codes to enhance decoding performance, eliminating the need for serial scheduling or post-processing. Our focus is on two-block Calderbank-Shor-Steane (CSS) codes, which are composed of symmetric stabilizers that hinder the performance of conventional iterative decoders with uniform update rules. Specifically, our analysis shows that, under the isolation assumption, the min-sum decoder fails to converge when constant-weight errors are applied to symmetric stabilizers, as variable-to-check messages oscillate in every iteration. To address this, we introduce a decoding technique that exploits this oscillatory property by applying distinct update rules: variable nodes in one block utilize messages from previous iterations, while those in the other block are updated conventionally. Logical error-rate results demonstrate that the proposed decoder significantly outperforms the normalized min-sum decoder and achieves competitive performance with belief propagation enhanced by order-zero ordered statistics decoding, all while maintaining linear complexity in the code's block length.
\end{abstract}	

\section{Introduction}
\IEEEPARstart{
Q}{uantum} low-density parity-check (QLDPC) codes have gained attention over topological codes because of their superiority in the minimum distance and their code rate scaling asymptotically~\cite{panteleev2021quantumLinearMinDLocalTestable, QuantumTannerCodeszemor}. The QLDPC codes families, such as the bivariate bicycle (BB) codes, are numerically proven to be more efficient than topological codes for quantum error-correction (QEC)~\cite{bravyi2024high}. 
However, the quantum decoding problem remains challenging for QLDPC codes. Unlike classical error-correction, in QEC, the goal lies in finding the most probable coset of degenerate errors that matches a given error syndrome; hence there can be multiple equally likely error patterns that match the input syndrome. Therefore, in principle, in QEC, the decoder is not restricted to converge to a single error pattern, as any error pattern belonging to the same coset is a valid error estimate. This phenomenon is known as \textit{degeneracy}. Although degeneracy is not inherently harmful, it hinders the performance of message-passing decoders that are efficient in the classical domain, such as belief-propagation (BP) or min-sum (MS) decoding. In fact, when the minimum distance of the code becomes significantly larger than the weight of the stabilizers, meaning that the number of degenerate errors increases, the performance of the message-passing decoder deteriorates. 

To mitigate the impact of degeneracy, solutions such as modifying the message-passing rules using post-processing or changing the message scheduling  have been proposed. Both probabilistic and deterministic modifications of BP~\cite{Poulin_2008,BPOTS} demonstrate improved performance on highly degenerate codes. Additionally, scheduling strategies such as layered BP and MS algorithms improve decoding by leveraging degeneracy, although with added latency~\cite{layered,informed}. Post-processing techniques such as Ordered Statistics Decoding (OSD)~\cite{osd}, Stabilizer Inactivation (SI)~\cite{StabInactivation_Julien_2022}, and BP with guided decimation (BPGD)~\cite{decim} offer further performance gains at the cost of increased complexity (e.g. $O(n^3)$, $O(n^2\log n)$, $O(n^2)$, where $n$ is the block length).
Recently, additional post-processing techniques such as~\cite{amb,local}, which are based on more efficient handling of the matrix inversion problem of BP-OSD, have been proposed. Ambiguity clustering~\cite{amb} dynamically builds a block or cluster structure in the parity check matrix. In this way, the decoder handles ambiguities by focusing on local error regions rather than trying to solve the entire system at once. On a similar note, localized statistics decoding~\cite{local} independently and concurrently processes disconnected subgraphs (clusters) where errors typically occur, although complexity remains cubic with respect to the cluster size. Finally, recent works implemented post-processing techniques targeting more accurate noise models such as BP plus ordered Tanner forest (BP+OTF)~\cite{otf}, sliding window BPGD decoding~\cite{window} and SymBreak decoder~\cite{symbreak}.

In addition to the degeneracy property that is inherent in all good quantum codes,  QLDPC codes also require long-range qubit connections, which induce additional delays\cite{xu2024constant, bravyi2024high},  
and further restrict the latency budget available for decoding implementations. Therefore, the main challenge in QEC is to balance the decoding performance without the penalty of high decoding latency. 

MS decoding is a low complexity variant of BP decoding which is widely considered in classical error correction. In particular, normalized min-sum (nMS) provides a good approximation of BP decoding if the normalization parameter is carefully selected~\cite{05CDEFH}. Syndrome-based MS can be applied for QEC, but good results are only obtained when one applies post-processing techniques like OSD and SI or when serial/adaptive scheduling is used~\cite{memory,informed}. The failures of MS and any iterative decoder with symmetric rules are linked to the existence of symmetric stabilizers which impose degeneracy~\cite{quantumTS,Fuentes_DegeneracyImpact_IEEEAccess21}.

In this work, we focus on two-block Calderbank-Shor-Steane (CSS) codes. Using a computation tree analysis, we prove that the variable-to-check messages exchanged within a symmetric stabilizer oscillate during each iteration of MS decoding. This oscillation occurs specifically when constant-weight error patterns are applied to the symmetric stabilizers, which are uncorrectable by MS. 
The proposed decoder, namely \emph{MS with past influence or MS-PI decoder}, leverages the oscillatory nature of the messages by introducing a memory attribute, i.e., by employing a rule that utilizes past messages during the variable update rule, and is thus able to correct these constant weight error-patterns. Similar approaches where the influence of past messages is utilized have also been explored in classical error-correction literature, i.e., the self-corrected MS~\cite{savin08isit}. MS-PI incorporates two (or more) variable update rules, using the block matrices as a distinguishing parameter between them.
More specifically, 
when the sign of a variable-to-check message does not change in two subsequent iterations, we employ the traditional MS update rule (see Section~\ref{sec:MS}), while if its sign is flipped, the new message is computed by summing the current and previous message. 
By doing so, we introduce asymmetry in the decoding process, as the variable-to-check messages in one block matrix differ from those in the other block matrix, in contrast to the symmetric behavior of the traditional MS decoder 
We present the algorithm and the intuition behind along with detailed analysis in Section~\ref{sec:memMS}. Simulation results demonstrate that by exploiting the oscillations of variable-to-check messages and at the same time by introducing different rules with respect to the block matrix, our decoder has significantly improved its performance compared to nMS. Also nMS-PI slightly outperforms BP-OSD decoding (which is used as a benchmark for performance), in only $50$ iterations, while also maintaining similar complexity to nMS. In particular, at low crossover probabilities, nMS-PI achieves increasingly larger performance gains over nMS as the blocklength (and thus degeneracy) of BB codes grows. For instance, for a crossover probability $\alpha = 0.02$, nMS-PI reduces the logical error rate by three orders of magnitude for the $[[288,12,18]]$ code and by one order of magnitude for the $[[144,12,12]]$ code.
Notably, nMS-PI achieves a threshold of $7.8 \%$ for the family of BB codes when run for $50$ iterations. Threshold further increases to $8 \%$ and $8.1 \%$ when the number of iterations is extended to $100$ and $200$, respectively.

The rest of the paper is organized as follows. In Section~\ref{sec:pre}, we introduce the preliminaries of QEC along with an overview of syndrome MS decoding. In Section~\ref{sec:MScomp}, we analyze the behavior of MS for symmetric stabilizers. In Section~\ref{sec:memMS}, we give the intuition behind the proposed algorithm by providing analysis based on computation trees and we finally present the proposed decoder. Finally, Section~\ref{sec:performance}, outlines simulation results obtained over BB codes.

\section{Preliminaries}
\label{sec:pre}
\subsection{Stabilizer Formalism}
Let \(\mathbb{F}_2^n\) denote the field of binary vectors of length \(n\). The \textit{Hamming weight} of an element in \(\mathbb{F}_2^n\) is defined as the number of its non-zero entries. An \([n, k, d]\) linear code \(C \subset \mathbb{F}_2^n\) is a linear subspace of \(\mathbb{F}_2^n\) spanned by \(k\) independent vectors, such that each codeword in \(C\) has a minimum Hamming weight of at least \(d\) (minimum distance of the code). A code \(C\) can also be represented by a parity-check matrix \(H\) of size \((n-k) \times n\), where \(C = \ker H\). Throughout this study, we consider regular parity-check matrices, where $d_v$ and $d_c$ denote the column-weight (variable degree) and row-weight (check degree) of $H$ respectively.

Let \((\mathbb{C}^2)^{\otimes n}\) denote the \(n\)-dimensional Hilbert space, and \(P_n\) the \(n\)-qubit Pauli group. A \textit{stabilizer} group is defined as an Abelian subgroup \(S \subset P_n\). An \(\llbracket n,k,d \rrbracket\) stabilizer code is a \(2^k\)-dimensional subspace \(\mathcal{C} \subset (\mathbb{C}^2)^{\otimes n}\) that satisfies 
\[
s_i\ket{\Psi} = \ket{\Psi}, \quad \forall\ s_i \in S, \ \ket{\Psi} \in \mathcal{C}.
\]
A \(\llbracket n, k_X - k_Z^{\perp}, d \rrbracket\) CSS code \(\mathcal{C}\) is a type of stabilizer code constructed using two classical codes, \(C_X = \ker H_X\) and \(C_Z = \ker H_Z\), with parameters \([n,k_X,d_X]\) and \([n,k_Z,d_Z]\), respectively, such that \(C_Z^{\perp} \subset C_X\) and \(C_X^{\perp} \subset C_Z\)~\cite{calderbank1996quantum_exists}. The minimum distance of the code is \(d \geq \min\{d_X, d_Z\}\), where \(d_X\) is the minimum Hamming weight of a codeword in \(C_X \setminus C_Z^{\perp}\), and \(d_Z\) is the minimum Hamming weight of a codeword in \(C_Z \setminus C_X^{\perp}\). A detailed description of the stabilizer formalism can be found in~\cite{Gottesman97}.

In this work, we focus on two-block CSS codes, which include bicycle QLDPC codes~\cite{mackay_quantum}, BB codes~\cite{bravyi2024high} and two-block group algebra codes~\cite{PhysRevA.109.022407} as specific instances. To briefly describe this family, consider two binary \(n \times n\) matrices \(A\) and \(B\) (blocks) that commute, i.e., \(AB = BA\). The parity-check matrices are then defined as:
\begin{equation}
H_X = [A, B], \quad H_Z = [B^T, A^T].
\label{eq:pcm}
\end{equation} 
It follows that \(H_X H_Z^T = AB + BA = 0\), ensuring that the commutativity condition is satisfied, resulting in a CSS code. As an example, binary circulant matrices \(A\) and \(B\) can be used, as they always commute. This class of codes encompasses the bicycle codes from~\cite{mackay_quantum} as a special case when \(B = A^T\).

CSS codes facilitate binary decoding because of their structure, i.e., $X$ errors are decoded using $H_Z$, and $Z$ errors are decoded using $H_X$. Therefore, one can consider the two independent binary symmetric channels (BSCs) rather than the
depolarizing channel, thus ignoring the correlation between the $X$ and $Z$ errors~\cite{mackay_quantum}. Let us denote as $\mathbf{e}=[\mathbf{e}_X, \mathbf{e}_Z]$ the binary representation
of a Pauli error acting on the $n$ qubits. The corresponding input syndromes are obtained as $\mathbf{s}_Z= \mathbf{e}_X  H_Z^T$
and $\mathbf{s}_X= \mathbf{e}_Z  H_X^T$, or in a compact form, $\mathbf{s}=[\mathbf{s}_X, \mathbf{s}_Z]$. 

A zero syndrome vector indicates that all stabilizers commute with the error pattern, implying no detectable errors. Conversely, a non-zero syndrome signals the presence of an error. The goal of a syndrome-based decoder is to estimate an error pattern, $\hat{\mathbf{e}}$, which produces a syndrome $\mathbf{\hat{s}}$ that matches $\mathbf{s}$. Decoding is  successful if the initial error $\mathbf{e}$ is recovered up to a stabilizer, which means that $\mathbf{e} 
 \xor \hat{\mathbf{e}}$ belongs to the rowspace of the parity check matrix $H$. Error correction fails when the input syndrome $\mathbf{s}$ has not been matched, or when decoding results in a \textit{logical error}, 
such that $\mathbf{e} \oplus \hat{\mathbf{e}}$  commutes with all the stabilizers, but it's not in the rowspace of $H$.

The $H_X$ and $H_Z$ stabilizer matrices can each be represented as a \textit{Tanner graph},  which is a bipartite graph with two sets of nodes:  $n$ variable (qubit) nodes denoted by $j \in \{1, 2, ...,n\}$ and  $m$ check nodes denoted by $i \in \{1, 2, ...,m\}$. If the $(i,j)$-th entry of the binary parity check matrix is not zero, then there is an edge connecting the variable node $j$ and check node $i$ (edges connect neighbouring nodes).
We denote by $\mathcal{M}(j)$ the indices of the neighboring check nodes of the variable node $j$, and by $\mathcal{N}(i)$ the indices of the neighboring variable nodes of the check node $i$.
For the rest of this paper, only one type of error ($X$ error) and its decoding will be considered; without loss of generality, the notation $H$ will refer to $H_Z$, $\mathbf{e}$ will refer to $\mathbf{e}_X$, and $\mathbf{s}$ will refer to $\mathbf{s}_Z$.

\subsection{Syndrome based MS decoding}
\label{sec:MS}
Message-passing algorithms can be employed for QLDPC decoding in a similar fashion to the classical LDPC case. The main difference is that the QLDPC decoders take
as input the syndrome obtained after the stabilizer measurements, whereas LDPC decoders take as input a noisy
version of the codeword. The QLDPC decoder aims to
output an error pattern whose syndrome matches the measured
syndrome. 

We will now briefly describe the syndrome-based nMS decoder. A check-to-variable message sent from check node $i$ to variable node $j$ at the $\ell$-th iteration is denoted by $\mu^{(\ell)}_{i,j}$, whereas a variable-to-check message sent from variable node $j$ to check node $i$ at the $\ell$-th iteration is denoted by $\nu^{(\ell)}_{j,i}$. Assuming a noise-free syndrome measurement and a binary symmetric channel inducing only $X$ errors with probability $\alpha$, the measured syndrome value is $s_i \in \{0,1\}$, $i \in \{1, 2, ..., m\}$ and the a priori log-likelihood ratio value for every variable node is given by $\lambda_j=\lambda=\log({\frac{1-\alpha}{\alpha}})$. For the sake of brevity, we denote the natural numbers from  $1$ to $n$ by $[n]$.
    
The variable node update rule is computed as: 
\begin{equation}
    \nu^{(\ell)}_{j,i} = \lambda_j + \sum\limits_{i' \in \mathcal{M}(j) \backslash \{ i \}}\mu^{(\ell)}_{i',j},
    \label{eq:vnu}
\end{equation}
and the check node update rule is computed as: 
\begin{equation}
   \mu^{(\ell)}_{i,j} = (1-2s_i)\prod\limits_{j' \in \mathcal{N}(i) \backslash \{j\}} \mathrm{sgn}(\nu^{(\ell-1)}_{j',i})\min\limits_{j' \in \mathcal{N}(i) \backslash \{j\}}|\nu^{(\ell-1)}_{j',i}|,
    \label{eq:cnu}
\end{equation}
where the sign function is defined as 
\[\mathrm{sgn}(u)= 
\begin{cases}
        -1, & \text{if $u<0$}\\
   +1, & \text{otherwise.}
\end{cases}\]
In this work we use the normalized version of MS algorithm, hence check-to-variable messages are multiplied by a scalar $\beta$. Finally, the error estimate $\hat{\mathbf{e}}$ is given by the decision update function:
\begin{equation}
\hat{e}_j=\mathrm{sgn}\left(\lambda_j + \sum\limits_{i \in \mathcal{M}(j) }\mu^{(\ell)}_{i,j}\right)
\label{eq:dec}
\end{equation}
The decoding procedure is continued until the maximum number of iterations $L$ has been reached or until the syndrome at the $\ell$-th iteration equals the input syndrome. The decoding succeeds when the estimated error pattern $\hat{\mathbf{e}}$ is equal to the
actual error pattern, or one of the equivalent error patterns that
lead to the same syndrome (degeneracy).

Recalling the quantum decoding problem, degenerate errors which are harmful to iterative decoders are referred to as symmetric degenerate errors. The corresponding sets of variable nodes are referred to as symmetric stabilizers, or simply  $(w, 0)$ trapping sets~\cite{quantumTS}, where $w$ is the weight of the stabilizer (number of variable nodes) and $0$ is the number of odd-degree check nodes in the set. 
Trapping sets are thoroughly defined and discussed in~\cite{ontology} and~\cite{quantumTS}. Iterative decoders with symmetric rules typically fail when $\frac{w}{2}$-weight error-patterns are applied to $(w,0)$ trapping sets, which is the case that we will analyze next.
\section{Computation Tree Analysis}
\label{sec:MScomp}

We will now show how MS decoding evolves for every $\frac{w}{2}$ error-pattern inside a $(w,0)$ stabilizer, using \emph{computation trees}, a tool also deployed in classical error correction to analyze iterative decoding~\cite{96Wiberg,SiegelComp}. Computation trees have also been used in QEC~\cite{blind} in order to demonstrate error-correction capabilities of the MS decoder for toric codes. A $k$-iteration computation tree unrolls the Tanner graph into a tree structure rooted at a variable node, recursively adding edges and leaf nodes that participate in the iterative message-passing decoding over $k$-iterations. 
\begin{figure}[ht]
 \centering
\begin{subfigure}[]{
 \centering
\includegraphics[width=0.16\textwidth]{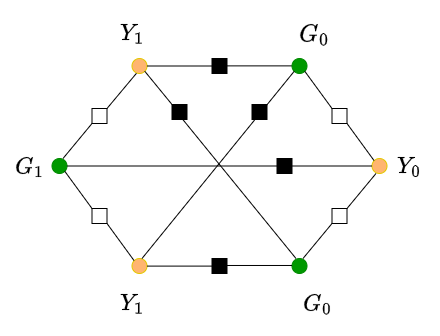}
\label{fig:comparisonCompTree1}
}
\end{subfigure}
\begin{subfigure}[]{
 \centering
\includegraphics[width=0.26\textwidth]{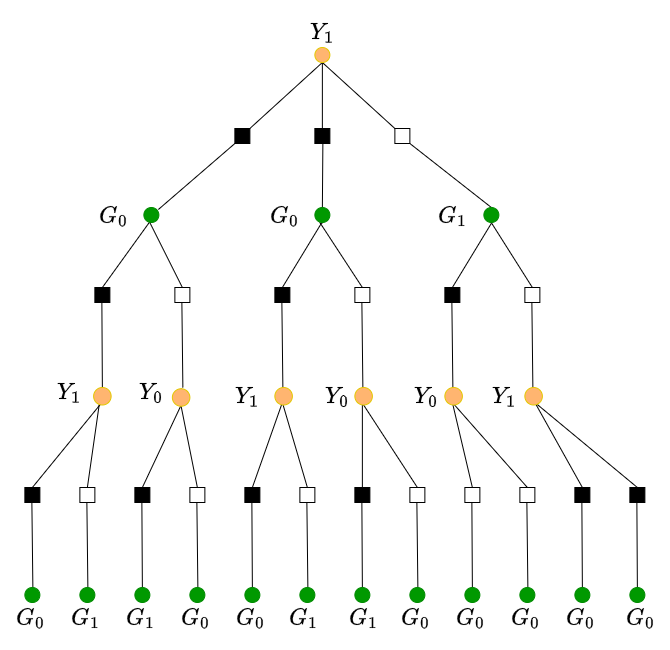}
\label{fig:comparisonCompTree2}
}
\end{subfigure}
\caption{A $(6,0)$ stabilizer of a two-block CSS code with $d_v=3$ and $d_c=6$
and its corresponding computation tree which is unrolled for three iterations. Variable nodes are colored yellow or green, depending on the block to which they belong. In this example, a weight-$3$ error pattern is applied to two variable nodes of yellow color and one variable node of green color. To indicate their respective error states,  the variable nodes of yellow color are labeled \( Y_1 \) if they are in error and \( Y_0 \) if they are error-free, while green nodes are labeled \( G_1 \) and \( G_0 \). Note that any weight-$3$ error-pattern applied to two yellow and one green node would result in the computation tree of Fig.~\ref{fig:comparisonCompTree}\subref{fig:comparisonCompTree2}.
}
\label{fig:comparisonCompTree}
\end{figure}
Throughout our analysis we will assume $(w,0)$ symmetric stabilizers, where $w=d_c$ is even and $w\geq 4$. 
We also assume that $X$ errors are induced by a BSC with crossover probability $\alpha$. 
Let $\mathcal{V}$ denote the trapping set. As an example, consider a trapping set with six variable nodes, and its induced graph shown in Fig.~\ref{fig:comparisonCompTree}\subref{fig:comparisonCompTree1}  (conventionally, but with moderate abuse of terminology, we refer to this induced graph also as a trapping set). In Fig.~\ref{fig:comparisonCompTree}, we color the circles with yellow and green to indicate variable nodes corresponding to the columns of the block matrices \( A \) and \( B \), respectively~(\ref{eq:pcm}). Therefore, the set of variable nodes \( \mathcal{V} \) is partitioned into two disjoint subsets: \( \mathcal{V}_{\text{yellow}} \), representing yellow nodes, and \( \mathcal{V}_{\text{green}} \), representing green nodes. Each yellow node \( j \in \mathcal{V}_{\text{yellow}} \) is assigned a binary labeling function \( Y: \mathcal{V}_{\text{yellow}} \to \{Y_1, Y_0\} \):
\[
Y(j) = 
\begin{cases} 
Y_1, & \text{if } j \text{ is initially in error,} \\
Y_0, & \text{if } j \text{ is initially error-free.}
\end{cases}
\]
Hence, \( Y_1 \) denotes the state ``initially in error", and \( Y_0 \) denotes the state ``initially error-free". The binary labels for green nodes (\( \mathcal{V}_{\text{green}} \)) are also defined in the same manner. As an example, in Fig.~\ref{fig:comparisonCompTree}, a weight-$3$ error-pattern is applied to two ``yellow" nodes and one ``green" node, hence there are two nodes with labels $Y_1$, one node with label $G_1$ and the rest of the nodes are labeled as $Y_0$ or $G_0$ depending on their color.

For the upcoming analysis, we will adopt the \emph{isolation assumption}~\cite{iso}.
This assumption states that all nodes outside the trapping set have converged, meaning that the amplitude of their variable-to-check messages is large enough so that they do not propagate into the trapping set under MS decoding. This implies that for each check node, it is sufficient to consider only the edges connected to variable nodes inside the trapping set, making such check nodes effectively of degree two in the computation tree. This will be useful in analyzing how the variable-to-check messages inside the trapping set evolve. Specifically, the degree-$2$ checks simplify the computation of check-to-variable messages, as the minimum operation is unnecessary. Instead, the variable-to-check messages can be expressed as a recursive function of the variable-to-check messages from the previous iteration. For simplifying the following analysis we will use different notation for the variable-to-check messages.
In the next paragraphs we will prove the following Theorem:
\begin{theorem}
    Consider a $(w,0)$ trapping set $\mathcal{V}$, with $w\geq 4$. Let $\frac{w}{2}$ variable nodes of $\mathcal{V}$ be initially in error. Then, under the isolation assumption, the variable-to-check messages exchanged within the trapping set, under the MS algorithm, will oscillate at every iteration \( \ell \).
    \label{theorem}
\end{theorem}
\begin{proof}
Theorem~\ref{theorem} will be proven by considering the two following cases for which a weight-$\frac{w}{2}$ error-pattern can be 
applied to a 
symmetric stabilizer. 
\begin{enumerate}
\item \emph{$\frac{w}{2}$ errors applied to nodes of only one color}:
Without loss of generality, assume that $\frac{w}{2}$ errors are applied to yellow nodes exclusively. In such case  
the label $Y_1$ is assigned to all variable nodes and all checks are unsatisfied. This results to only one type of variable-to-check message exchanged within the trapping set, so unrolling the computation tree yields a single equation describing this message, given by:
\begin{equation}
\begin{aligned}
a(\ell) &= -(\frac{w}{2}-1) \cdot a(\ell-1) +\lambda, \\
\end{aligned}
\label{eq:oneC}
\end{equation}
where $a(\ell)$ corresponds to a variable-to-check message sent from a variable node $Y_1$ to an unsatisfied check at iteration $\ell$ and is also illustrated in Fig.~\ref{fig:oneC}. The recursion in~(\ref{eq:oneC}) is derived by simplifying~(\ref{eq:vnu}), considering that only one incoming check-to-variable message contributes to the computation.
For simplicity, we will consider the homogeneous recurrence relation:
\begin{equation}
a(\ell) = -(\frac{w}{2}-1) \cdot a(\ell-1),
\end{equation}
for which, the solution is:
\begin{equation}
a(\ell) = \lambda \cdot (-\frac{w}{2}+1)^{\ell}.
\end{equation}

Hence we prove that variable-to-check messages oscillate in each iteration and their amplitude grows exponentially depending on the weight of the stabilizer. Finally, the error estimate for every node is given by (subscripts correspond to labels of variable nodes):
\begin{equation}
\begin{aligned}
\hat{q}_{Y_1} &=\hat{q}_{G_0}= -\frac{w}{2}\cdot a(\ell-1)=\lambda\cdot (-\frac{w}{2}+1)^{\ell}\cdot \frac{-w/2}{-w/2+1}, 
\end{aligned}
\end{equation}
which shows that the hard decision of MS will oscillate ($\ell$ at the exponent) between estimating the all-zero error-pattern and all-one weight-$w$ error-pattern.

\item \emph{$\frac{w}{2}$ errors are applied to  nodes of both colors}:
Assume that $f$ errors are applied to yellow nodes and $g$ errors are applied to green nodes. It holds that $f+g=\frac{w}{2}$, where  $f \geq g \geq 1$. For this case $Y_1$ is equivalent to $G_0$ and $Y_0$ is equivalent to $G_1$. Therefore, all possible \emph{transitions} in the computation tree are summarized in Fig.~\ref{fig:trn}. Here, a transition refers to a change or preservation of labels between two neighboring variable nodes. Specifically, a satisfied check switches the label, while an unsatisfied check retains the same label. This limited set of transitions 
enables us to develop recursive formulas for the variable-to-check messages of the stabilizer without the need to construct the whole computation tree. Since variable nodes are labeled by  $Y_1$ and $G_1$, and check nodes can either be there are four distinct types of variable-to-check messages that can be transmitted within the symmetric stabilizer, as shown in Fig.~\ref{fig:subtree}. Each message, i.e., $a(\ell),b(\ell),c(\ell),d(\ell)$, can be described by a recursive formula which is derived by the transitions of Fig.~\ref{fig:trn}.
For instance, in the case of $a(\ell)$, a variable node labelled as $Y_1$ is connected to $f$ unsatisfied checks and $g-1$ satisfied checks. The unsatisfied checks preserve the label of the node and therefore propagate the $b(\ell-1)$ message, $f$ times in total, while the satisfied checks switch the label of the node and therefore propagate the $c(\ell-1)$ message, $g-1$ times in total. The rest of variable-to-check messages are computed in a similar manner, giving rise to a system of four linear recursions:
\begin{equation}
\begin{aligned}
a(\ell) &= (g - 1) \cdot c(\ell-1) - f \cdot b(\ell-1)+\lambda, \\
b(\ell) &= g \cdot c(\ell-1) - (f - 1) \cdot b(\ell-1)+\lambda, \\
c(\ell) &= (f - 1) \cdot a(\ell-1) - g \cdot d(\ell-1)+\lambda, \\
d(\ell) &= f \cdot a(\ell-1) - (g - 1) \cdot d(\ell-1)+\lambda.
\end{aligned}
\end{equation}
By solving the homogeneous system, we get the following closed-form expression in terms of eigenvalues:
\begin{equation}
\begin{aligned}
\begin{pmatrix}
a(\ell) \\
b(\ell) \\
c(\ell) \\
d(\ell)
\end{pmatrix}
&=
D_1 \mathbf{v}_1 \cdot 1^{\ell} + D_2 \mathbf{v}_2 \cdot \left(i \sqrt{\frac{w}{2} - 1}\right)^{\ell} \\
&\quad + D_3 \mathbf{v}_3 \cdot \left(-i \sqrt{\frac{w}{2} - 1}\right)^{\ell}+ D_4 \mathbf{v}_4 \cdot \left(1 - \frac{w}{2}\right)^{\ell},
\end{aligned}
\end{equation}
where:
\( D_1 \), \( D_2 \), \( D_3 \), and \( D_4 \) are constants determined by the initial conditions and \( \mathbf{v}_i \) represents the eigenvector corresponding to the eigenvalue \( \delta_i \), for \( i = 1, 2, 3, 4 \). For \( w \geq 4 \), the eigenvalue \( \delta_1 = 1 \) contributes a constant component to the solution, while the purely imaginary eigenvalues \( \delta_2 = i \sqrt{\frac{w}{2} - 1} \) and \( \delta_3 = -i \sqrt{\frac{w}{2} - 1} \) introduce oscillatory behavior with a frequency proportional to \( \sqrt{\frac{w}{2} - 1} \). The real eigenvalue \( \delta_4 = 1 - \frac{w}{2} \), which satisfies \( \delta_4 \leq -1 \), results in \emph{exponential divergence with alternating sign}, leading to growth in magnitude while crossing zero. 
The error estimates for $Y_1$ and $G_0$ nodes coincide, therefore MS cannot achieve convergence (same with $Y_0$ and $G_1$). 
\end{enumerate}
\end{proof}
\begin{figure}[ht]
    \centering
    \includegraphics[width=0.12\textwidth]{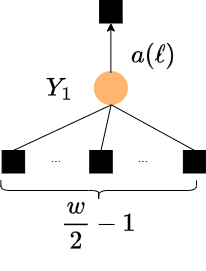}
    \caption{The only type of variable-to-check messages exchanged in a $(w,0)$ trapping set at iteration $\ell$, when $\frac{w}{2}$ errors are applied to yellow variable nodes.}
    \label{fig:oneC}
\end{figure}
\begin{figure}[ht]
    \centering
    \includegraphics[width=0.22\textwidth]{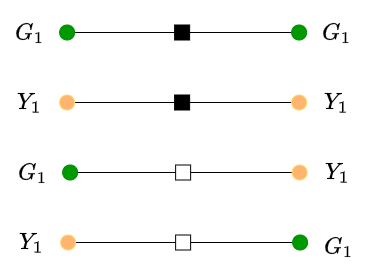}
    \caption{
    All possible transitions (change or preservation of labels
of variable nodes) between nodes of a symmetric stabilizer when errors are applied to nodes of both colors. Unsatisfied checks retain the same label between neighboring nodes (either $G_1$ or $Y_1$) while satisfied checks switch the label. This limited set of transitions reduces the computation tree analysis to a system of four recursive equations.} 
    \label{fig:trn}
\end{figure}
\begin{figure}[ht]
    \centering
 \includegraphics[width=0.4\textwidth]{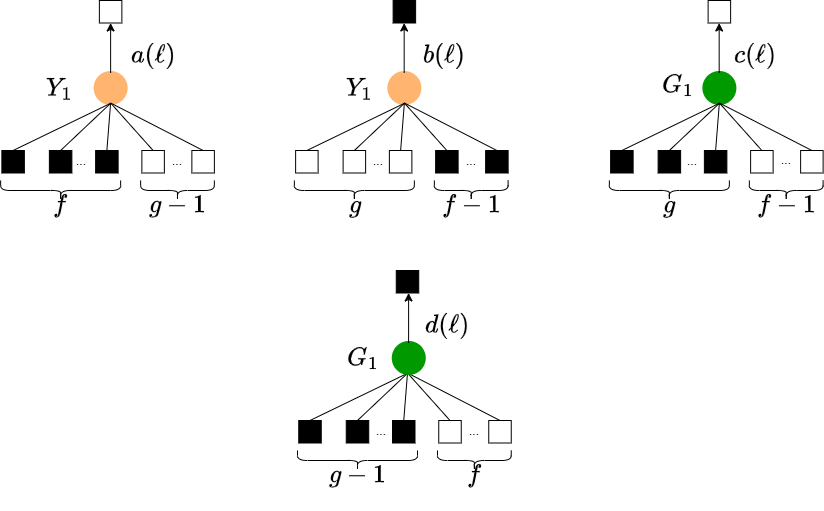}
    \caption{The four possible types of variable-to-check messages exchanged in a $(w,0)$ symmetric stabilizer when errors are applied to nodes of both colors. Note that $f+g=\frac{w}{2}$.}
    \label{fig:subtree}
\end{figure}
We note that the theorem and its proof can be extended to nMS, with the key difference being that the updated recursions incorporate the scalar $\beta$, which scales the check-to-variable messages through multiplication. Finally, this theorem is extended in the next Section for the analysis of MS-PI, and we complement the theoretical analysis with numerical results that corroborate the theorem, as the variable-to-check messages introduce non-linearity.
\section{Min-Sum with Past Influence}
\label{sec:memMS}
\subsection{Exploiting Oscillations of Variable-to-Check Messages}
For MS decoder, we proved that the variable-to-check messages exchanged within any $(w,0)$ symmetric stabilizer oscillate in every iteration, when any weight-$\frac{w}{2}$ error-pattern is applied to the stabilizer. This can be exploited by tracking oscillating messages and then modifying the update rules of the particular nodes.
Intuitively, since symmetric stabilizers are only formed by interconnecting the two block matrices, such erroneous configurations can be dealt with by assigning different update rules to variable nodes regarding which block matrix they belong to. 
Therefore, we propose  handling the oscillations of variable-to-check messages only for variable nodes of one block matrix. Such an action induces asymmetry and helps the iterative decoder to converge.

For the two-block codes, we will deploy two variable update rules, each applied to the block matrices in~(\ref{eq:pcm}). One rule will be the same as the MS rule in~(\ref{eq:vnu}).
For the second rule, it is crucial to use a rule that does not negatively affect performance and at the same time, introduces asymmetry when used in conjunction with the first rule. Since uncorrectable errors are characterized by oscillating messages, we will exploit that to our advantage. In particular, at each iteration $\ell$ we will compare the sign of the current variable-to-check message to the sign of the variable-to-check message sent during the previous iteration $\ell-1$ across the same edge. If the two signs are equal, then the update will be the same as (\ref{eq:vnu}), otherwise the update will be equal to the summation of the current and the previous message.
On one hand, the summation introduces asymmetry inside the symmetric stabilizer. On the other hand, it reduces the amplitude of the oscillating messages, which are deemed unreliable. This reduction helps to avoid overestimating the variable-to-check messages.
This rule is effective, since the variable-to-check messages of $(w,0)$ trapping sets oscillate in every iteration as we proved in Theorem~\ref{theorem}.

MS with past influence or MS-PI decoder, is presented in Algorithm~\ref{alg:ms}. Our modification is outlined in the following equations and we assume it is applied to green nodes; however, it can be equally applied to yellow nodes without loss of generality. First, the temporary message is computed: \begin{equation}
\nu = \lambda_j + \sum\limits_{i' \in \mathcal{M}(j) \backslash \{ i \}}\mu^{(\ell)}_{i',j},
\label{eq:temp}
\end{equation}
and then, after sign flips are checked, the variable-to-check message is computed as follows:
\begin{equation}
\nu^{(\ell)}_{j,i} = \begin{cases}
\nu, & \text{if ($\mathrm{sgn}(\nu)=\mathrm{sgn}(\nu^{(\ell-1)}_{j,i}))$}\\
\nu+\nu^{(\ell-1)}_{j,i}, & \text{otherwise}.
\end{cases}
\label{eq:Pi}
\end{equation}
This modification can also be extended for other message-passing decoders.
\begin{algorithm}
    \caption{MS-PI decoder}
     \label{alg:ms}
  \textbf{Input}: $\alpha$, $\mathbf{s}, H$, $L$ 

  \textbf{Output}: $\hat{\mathbf{e}}$
  \begin{algorithmic}[1]
      \State $\ell= 0$ \Comment{Initialization}
      \State $\lambda_j = \log (\frac{1-\alpha}{\alpha} )$, $\forall j \in [n]$
      \State ${\hat{s}_i} = 0$, $\forall i \in [m]$ 
      \State $\nu^{(0)}_{j,i} = \lambda_j\ \forall\ i \in [m],\ j\in [n]$
      \While{($\ell < L$ $\land$ $\mathbf{\hat{s}}\neq \mathbf{s}$)}
        \State Check node update (\ref{eq:cnu})
                \For{$j \in \mathcal{V}_{\text{yellow}}$}
       \State Variable node update (\ref{eq:vnu})
   \EndFor
        \For{$j \in \mathcal{V}_{\text{green}}$}
        \State $\nu = \lambda_j + \sum\limits_{i' \in \mathcal{M}(j) \backslash \{ i \}}\mu^{(\ell)}_{i',j}$
        \State $\nu^{(\ell)}_{j,i} = \begin{cases}
\nu, & \text{if ($\mathrm{sgn}(\nu)=\mathrm{sgn}(\nu^{(\ell-1)}_{j,i}))$}\\
\nu+\nu^{(\ell-1)}_{j,i}, & \text{otherwise}.
		 \end{cases}$
   \EndFor
     \State Decision Update (\ref{eq:dec})
     \State $\mathbf{\hat{s}}=\hat{\mathbf{e}} \cdot H^T$
     \If {$\mathbf{\hat{s}}=\mathbf{s}$}  \Return $\hat{\mathbf{e}}$
     \EndIf
     \State $\ell = \ell +1$
      \EndWhile
    
  \end{algorithmic}
\end{algorithm}
\subsection{Numerical Analysis of MS-PI}
We will now analyze the behavior of MS-PI for weight-\( \frac{w}{2} \) error patterns applied to a \( (w, 0) \) trapping set under the isolation assumption. However, since the proposed decoder introduces non-linearity, it is hard to handle an analytical solution. Therefore, we will perform numerical analysis for the two cases that were previously studied in Theorem~\ref{theorem}. 
\begin{enumerate}
\item \emph{$\frac{w}{2}$ errors applied to nodes of only one color:}
\begin{figure}[!htb]
    \centering
\includegraphics[width=0.4\textwidth]{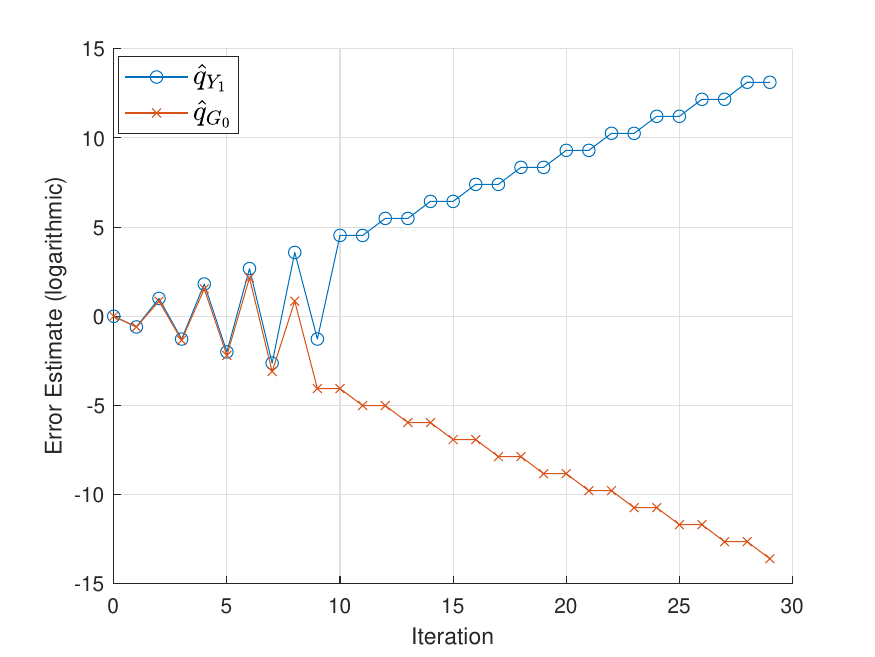}
    \caption{Numerical evaluation of error estimates for a weight-$3$ error applied to yellow nodes ($f=3,g=0$) of a $(6,0)$ stabilizer. Result is obtained by using the system of non-linear recursions in~(\ref{eq:Numone}) and~(\ref{eq:Numone1}). The scissor-like shape indicates that the decoder converges to one of the two degenerate error-patterns after the ninth iteration.}
    \label{fig:MSPIonecolor}
\end{figure}
Without loss of generality, we assume that all errors are applied to yellow nodes. The following non-linear system arises (for simplicity, we consider homogeneous relations):
\begin{equation}
\begin{aligned}
a_{\text{yellow}}(\ell) &= -(\frac{w}{2}-1) \cdot a_{\text{green}}(\ell-1) + \chi \cdot a_{\text{yellow}}(\ell-1), \\
a_{\text{green}}(\ell) &= -(\frac{w}{2}-1) \cdot a_{\text{yellow}}(\ell-1),
\end{aligned}
\label{eq:Numone}
\end{equation}
where subscripts \emph{yellow} and \emph{green} indicate the variable-to-check messages sent by yellow and green nodes respectively and \(\chi\) is an indicator function defined as follows:
\begin{equation}
\chi = \begin{cases} 
1 & \text{if } \mathrm{sgn}(-(\frac{w}{2}-1) \cdot a_{\text{green}}(\ell-1)) \neq \mathrm{sgn}(a_{\text{yellow}}(\ell-1)), \\
0 & \text{otherwise}.
\end{cases}
\label{eq:chi}
\end{equation}
The indicator function introduces a damping term to reduce oscillations when \( -(\frac{w}{2}-1)\cdot a_{\text{green}}(\ell-1) \) and \( a_{\text{yellow}}(\ell-1) \) have opposite signs. 
The error estimate for nodes $Y_1$ and $G_0$ is given as:
\begin{equation}
\begin{aligned}
\hat{q}_{Y_1} &= -\frac{w}{2}\cdot a_{\text{green}}(\ell-1)\\
\hat{q}_{G_0} &= -\frac{w}{2}\cdot a_{\text{yellow}}(\ell-1).
\end{aligned}
\label{eq:Numone1}
\end{equation}
In Fig.~\ref{fig:MSPIonecolor}, we numerically evaluated the non-linear system, with  $\frac{w}{2}=3$, and by initializing $a_{\text{yellow}}(0)$ and $a_{\text{green}}(0)$ to one. As illustrated, the decoder converges to one of the degenerate error-patterns.
\item \emph{$\frac{w}{2}$ errors are applied to  nodes of both colors:}
\begin{figure}[!htb]
    \centering
  \includegraphics[width=0.4\textwidth]{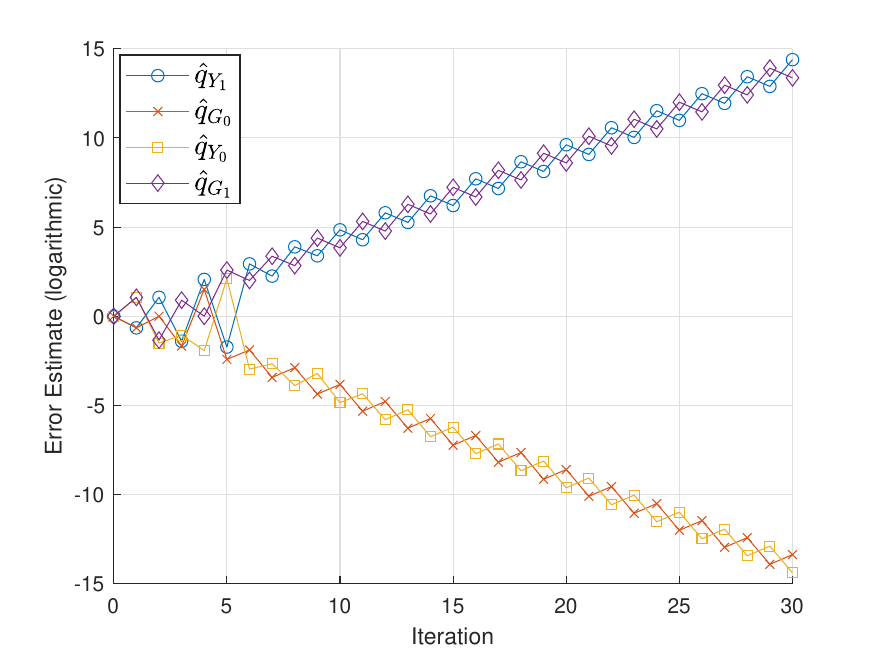}\caption{Numerical evaluation of error estimates for a weight-$4$ error-pattern ($f=3, g=1$) applied to both yellow and green nodes of a $(8,0)$ stabilizer. Result is obtained by using the system of non-linear recursions in~(\ref{eq:nonLin1}),~(\ref{eq:nonLin2}) and~(\ref{eq:nonLin3}). The decoder converges to one of the two degenerate error-patterns within six iterations.}
    \label{fig:evaluationMSPI}
\end{figure}
Now, since the new rule is only applied for variable nodes of yellow color, we have a system of eight recursions given below. Note that the indicator function compares the sign of the terms on its right to those on its left side, as in~(\ref{eq:chi}).
\begin{equation}
\begin{aligned}
a_{\text{yellow}}(\ell) &= (g - 1) \cdot c_{\text{green}}(\ell-1) - f \cdot b_{\text{green}}(\ell-1) \\
       &\quad + \chi \cdot a_{\text{yellow}}(\ell-1), \\
b_{\text{yellow}}(\ell) &= g \cdot c_{\text{green}}(\ell-1) - (f - 1) \cdot b_{\text{green}}(\ell-1) \\
       &\quad + \chi \cdot b_{\text{yellow}}(\ell-1), \\
c_{\text{yellow}}(\ell) &= (f - 1) \cdot a_{\text{green}}(\ell-1) - g \cdot d_{\text{green}}(\ell-1) \\
       &\quad + \chi \cdot c_{\text{yellow}}(\ell-1), \\
d_{\text{yellow}}(\ell) &= f \cdot a_{\text{green}}(\ell-1) - (g - 1) \cdot d_{\text{green}}(\ell-1) \\
       &\quad + \chi \cdot d_{\text{yellow}}(\ell-1),
\end{aligned}
\label{eq:nonLin1}
\end{equation}
\begin{equation}
\begin{aligned}
a_{\text{green}}(\ell) &= (g - 1) \cdot c_{\text{yellow}}(\ell-1) - f \cdot b_{\text{yellow}}(\ell-1), \\
b_{\text{green}}(\ell) &= g \cdot c_{\text{yellow}}(\ell-1) - (f - 1) \cdot b_{\text{yellow}}(\ell-1), \\
c_{\text{green}}(\ell) &= (f - 1) \cdot a_{\text{yellow}}(\ell-1) - g \cdot d_{\text{yellow}}(\ell-1), \\
d_{\text{green}}(\ell) &= f \cdot a_{\text{yellow}}(\ell-1) - (g - 1) \cdot d_{\text{yellow}}(\ell-1).
\end{aligned}
\label{eq:nonLin2}
\end{equation} 
Finally, the error estimates for nodes $Y_1$, $G_0$, $Y_0$ and $G_1$ are given as:
\begin{equation}
\begin{aligned}
\hat{q}_{Y_1} &= g\cdot c_{\text{green}}(\ell-1)-f\cdot b_{\text{green}}(\ell-1)\\
\hat{q}_{G_0} &= g\cdot c_{\text{yellow}}(\ell-1)-f\cdot b_{\text{yellow}}(\ell-1)\\
\hat{q}_{Y_0} &= f\cdot a_{\text{green}}(\ell-1)-g\cdot d_{\text{green}}(\ell-1)\\
\hat{q}_{G_1} &= f\cdot a_{\text{yellow}}(\ell-1)-g\cdot d_{\text{yellow}}(\ell-1).\\
\end{aligned}
\label{eq:nonLin3}
\end{equation}
MS-PI achieves convergence by forcing the estimates of $Y_1$ and $G_0$ to break their ties (same with $Y_0$ and $G_1$) which prevented the convergence in the MS algorithm.
The numerical evaluation of Fig.~\ref{fig:evaluationMSPI}, shows that for a $(8,0)$ stabilizer the estimate of nodes $Y_1$ and $G_1$ grows positively, while the estimates of  $Y_0$ and $G_0$ grow negatively, so the decoder converges to one of the two degenerate error-patterns. It is easy to see that if the modified rule is applied to green nodes, the decoder would converge to the other degenerate error-pattern.

\end{enumerate}

\section{Performance evaluation}
\label{sec:performance}
\begin{figure}[htbp]
    \centering
\includegraphics[width=0.4\textwidth]{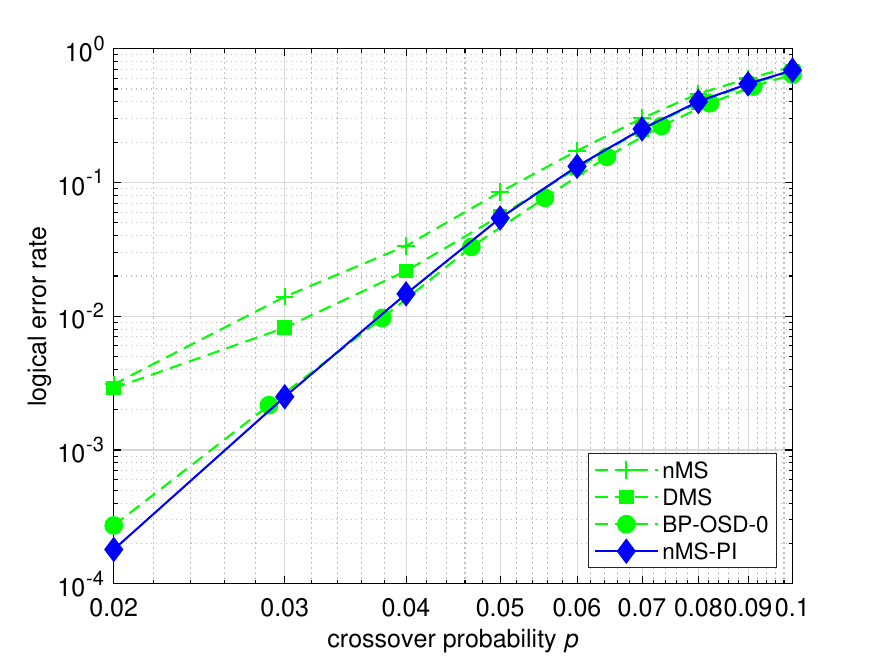}
    \caption{
    Performance of various decoders for the $[[144,12,12]]$ BB code. Notably, nMS-PI outperforms BP-OSD-$0$ in only $50$ iterations. Moreover, the second variable update rule only works well when used with the update rule of~(\ref{eq:vnu}).}

    \label{fig:memory144}
\end{figure}
\begin{figure}[htbp]
    \centering
\includegraphics[width=0.4\textwidth]{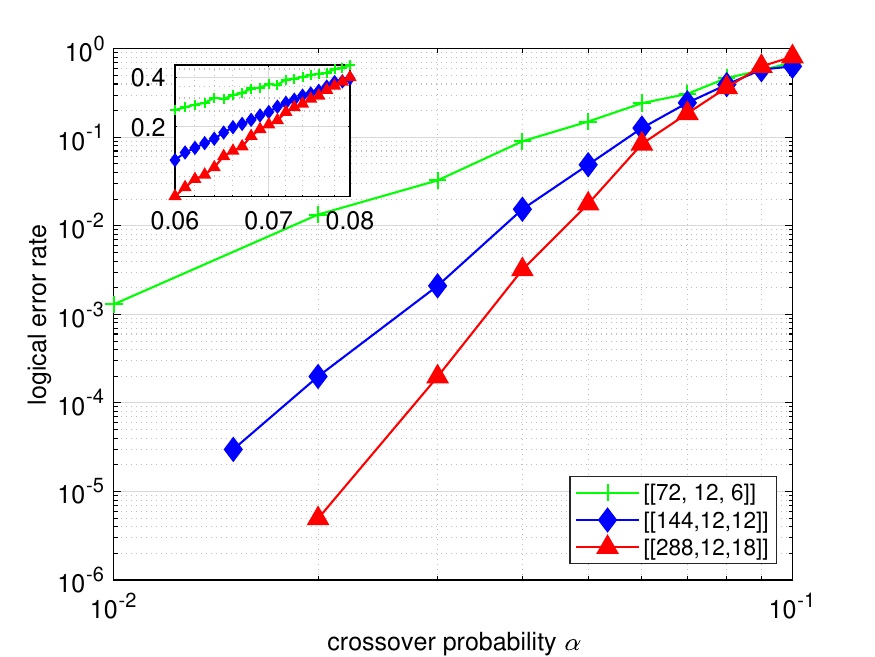}
    \caption{MS-PI performance for BB codes of various distances. The maximum iteration number is set to $L=50$ and the threshold approximates $7.8\%$.}
    \label{fig:BBthreshold}
\end{figure}
We simulate the performance of the proposed decoder for the family of BB codes. Unless noted otherwise, the maximum number of iterations for nMS ($\beta=0.875$) and nMS-PI is set to $50$ and every decoder shown in the plots follows parallel scheduling. The BP-OSD-$0$ decoder~\cite{osd}
is considered as the benchmark for evaluating the performance of nMS-PI. The plots for BP-OSD-$0$ decoder were generated based on~\cite{roffe_decoding_2020,Roffe_LDPC_Python_tools_2022}. A BSC with crossover probability $\alpha$ resulting in the bit-flip ($X$) errors is assumed for the noise model.

Fig.~\ref{fig:memory144} illustrates the performance of the proposed decoder, the two variants of nMS and BP-OSD-$0$ over the $[[144,12,12]]$ BB code. Damped MS (DMS)  refers to the nMS decoder where the variable update rule, as defined by~(\ref{eq:Pi}), is used uniformly to all variable nodes. We note that the proposed decoder significantly outperforms both instances of nMS and also slightly surpasses BP-OSD-$0$. The same plot also shows that there is a slight deviation in performance between the two nMS instances, indicating that there are error-patterns which are uncorrectable by one nMS instance but correctable by the other one. One can leverage that property by using diversity decoding, i.e., four decoders which apply all possible configurations of the two rules across the two block matrices (diversity decoding has been deployed in~\cite{collective}). 
The behavior of nMS-PI is consistent with BB codes of various lengths as observed in Fig.~\ref{fig:BBthreshold}, where the proposed decoder demonstrates a threshold which is approximately equal to $7.8 \%$.  The results attained for the family of BB codes for the BSC channel are comparable to the BP-OSD-$0$ results provided in~\cite{window}.
For low crossover probabilities, the performance gain of nMS-PI over nMS grows as the blocklength (and thus the degeneracy) of the BB codes increases. For instance, for the $[[288, 12, 18]]$ code and $\alpha=0.02$, nMS-PI achieves a three-order-of-magnitude reduction in logical error rate compared to nMS.
Finally, we observe that the threshold of nMS-PI improves with increasing iterations: for \( L = 100 \), the threshold rises to \( 8\% \), and for \( L = 200 \), it increases further to \( 8.1\% \).

\section{Conclusions and Future Work}
We have proposed a low-complexity message-passing decoding scheme, namely nMS-PI, applicable to QLDPC codes. 
Our decoder effectively utilizes degeneracy by applying distinct update rules - informed by past dynamics - to each block matrix. 
Simulation results demonstrate that our algorithm can significantly surpass nMS decoding performance and closely approach the performance of more complex decoders, such as BP-OSD, within a limited number of iterations. Furthermore, the results indicate that nMS-PI achieves a threshold of $7.8\%$ for the BB code family within just $50$ decoding iterations.
Notably, the proposed scheme demonstrates strong performance under parallel scheduling and exhibits linear complexity relative to the code's block length. Future research will focus on assessing the proposed decoder against more realistic noise models, such as the circuit-level noise model.

\bibliographystyle{IEEEtran}
\bibliography{totalRefs.bib}
\end{document}